\documentclass[reqno,11pt]{amsart}
\usepackage[english]{babel}
\usepackage[T1]{fontenc}
\usepackage{amsmath,amssymb,amsfonts,amsthm}
\usepackage{hyperref}
\usepackage{slashed}
\usepackage{graphicx,color}
\usepackage[bbgreekl]{mathbbol}
\usepackage{mathtools}
\usepackage[margin=1.5cm]{geometry}

\newtheorem{definition}{\textsf{Definition}}[section]
\newtheorem{theorem}[definition]{\textsf{Theorem}}
\newtheorem{proposition}[definition]{\textsf{Proposition}}

\bibstyle{apsrev4-1}

\newtheorem{remark}[definition]{\textsf{Remark}}

\newcommand{\qsp}[2]{\,\ensuremath{\raise.5ex\hbox{$#1$}\big\slash\raise-.5ex\hbox{$#2$}}} 

\newcommand{\pard}[2]{\frac{\delta#1}{\delta#2}}

\newcommand{\txi}[1]{\widetilde{\xi}^{#1}}
\newcommand{\tJ}[1]{\widetilde{J}_{#1}}
\newcommand{\tgd}[1]{\widetilde{g}^\dag{}^{#1}}

\newcommand{\tgam}{\widetilde{\gamma}}

\newcommand{\RR}{\mathbb{R}}

\begin{document}

\title[BV-BFV GR: Einstein-Hilbert Action]{BV-BFV approach to General Relativity: \\ Einstein-Hilbert action}

\author{A. S. Cattaneo}
\address{institut f\"ur Mathematik, Winterthurerstrasse 190, 8057 Z\"urich, Switzerland}
\email{alberto.cattaneo@math.uzh.ch}

\author{M. Schiavina}
\address{institut f\"ur Mathematik, Winterthurerstrasse 190, 8057 Z\"urich, Switzerland}
\email{michele.schiavina@math.uzh.ch}

\date{\today}

\thanks{This research was (partly) supported by the NCCR SwissMAP, funded by the Swiss National Science Foundation,  and by the 
COST Action MP1405 QSPACE, supported by COST (European Cooperation in Science and Technology). The authors acknowledge partial support of SNF Grants No. 200020-149150/1 and No. PDFMP2\_137103.
M.S. also acknowledges the Forschungskredit of the University of Z\"urich. 
The authors would like to thank J. Stasheff for his helpful comments, and P. Mn\"ev for relevant remarks on earlier stages of the work. Special thanks to the anonymous reviewer, for contributing to the exposition of the paper with precious comments and suggestions.}

\begin{abstract}The present paper shows that General Relativity in the ADM formalism admits a BV-BFV formulation. More precisely, for any $d+1\not=2$ (pseudo-)Riemannian manifold $M$ with spacelike or timelike boundary components, the BV data on the bulk induces compatible BFV data on the boundary. As a byproduct, the usual canonical formulation of General Relativity is recovered in a straightforward way. 
\end{abstract}

\maketitle

\section*{Introduction}
General Relativity (GR) in the Einstein-Hilbert formalism (EH) is not manifestly a gauge theory, in the sense that it is not a theory of principal connections on some space-time manifold. It is nevertheless a  field theory that admits a large group of transformations under which its action is invariant, namely the group of space-time diffeomorphisms. Generalised notions of gauge theories have been considered in the past where the requirement that the symmetry distribution be involutive has been relaxed to accommodate for more general cases like the Poisson sigma model or BF theories. The natural setting for a consistent treatment of such symmetries in view of (perturbative) quantisation is given by the Batalin-(Fradkin)-Vilkovisky \cite{BV81,BFV1,BFV2} (B(F)V) formalism, which generalises the well known and celebrated Becchi-Rouet-Stora, Tyutin \cite{BRST} (BRST) approach to the quantisation of gauge theories.

The BRST and B(F)V formalisms provide a cohomological description of gauge equivalent classes of fields and of coisotropic submanifolds \cite{stash97} (e.g. the critical locus of the action). The BFV formalism for the boundary can in principle be defined just in terms of the induced constraints and, under certain assumptions, it can be quantized to produce a complex whose cohomology in degree zero is the quantization of the reduced phase space. 

In the work of Cattaneo, Mn\"ev and Reshetikhin (CMR) \cite{CMR1,CMR2,CMR3}, it has been shown how it is possible, under certain assumptions, to induce a \emph{compatible}\/ BFV manifold structure on the boundary $\partial M$ starting from a BV manifold structure on the bulk space-time manifold $M$. The observation in CMR is that one needs a compatibility between the bulk BV structure and the boundary BFV structure for the state associated to the bulk to be a cocycle (and hence to define a physical state).  At the semiclassical level, the fundamental compatibility condition is that the failure of the bulk action to be the Hamiltonian function for the BV operator should be given by pullback of the boundary Noether $1$-form. 

This has several advantages, such as a compatible cutting-gluing procedure, which allows one to break topologically nontrivial manifolds into pieces for which one might expect the quantisation to be simpler, as well as a powerful understanding of the quantisation procedure itself as a suitable generalisation of the Atiyah-Segal axioms for (topological) gauge theories \cite{CMR3,Ati,Seg}. Moreover, this approach gives a handle on the classical theory, providing a clean understanding of the \emph{Hamiltonian} approach to classical gauge theories, after Dirac and his canonical constraint analysis \cite{Dirac}, which is in fact a first step towards their quantisation.

As a matter of fact, we will work out here the necessary conditions of the aforementioned BV-BFV approach to General Relativity, i.e. the starting point for the program of BV quantum gravity. Even though most of the results that we will present are chiefly a recovery of established knowledge (above all \cite{dWi} and subsequent citing literature), their straightforward derivation using the BV-BFV machinery is nontrivial and novel. Using a much weaker requirement than global hyperbolicity of the space-time manifold $M$ we are able to prove that the boundary theory is the symplectic reduction of the canonical Hamiltonian formalism of general relativity in the Arnowitt-Deser-Misner (ADM) formalism \cite{ADM}, even when symmetries are dynamically taken into account. The induced data, encoded in a boundary action, will encompass at the same time all the relevant Hamiltonian and momentum constraints, together with the structure of the residual gauge transformations on the boundary.

As a byproduct, we are able to address a recent question posed in \cite{BlohWein}, about the origin and nature of the constraint algebra, showing that it yields a nontrivial, non constant and yet linear coisotropic structure, not manifestly coming from a Lie algebra action. Indeed, the very compatibility of a BV structure in the bulk manifold with a BFV structure in the boundary means that the induced graded manifold given by the space of boundary fields  automatically represents the cohomological resolution of the phase space of the theory.

In Section \ref{Section:background} we will briefly review the ideas behind the BV-BFV formalism and the axioms one would like the gauge theory to satisfy in order to have access to the BV-BFV quantisation machinery \cite{CMR3}.

Further on we will show (Theorem \ref{Theo:ADM1}) that General Relativity in the Einstein-Hilbert formalism does indeed satisfy the axioms, for all pseudo-Riemannian manifolds of dimension $d+1\not=2$ with space/time-like boundaries, and for all Riemannian manifolds of dimension $d\not=2$. Notice that our result does not rule out the presence of null or type-changing boundaries, but the technique has to be adapted, and it will be subject of further research.

In Theorem \ref{Theo:ADMBSTR} we give explicit expressions in Darboux coordinates of the $(0)$-symplectic graded space of boundary fields. Its degree $0$ part is the well known symplectic space underlying the usual canonical description of Hamiltonian ADM GR, of which an explicit expression we give in Section \ref{Sect:ClassADM}. Again, the classical canonical analysis is strongly simplified when using the present approach, which substantially stems from the interpretation of the Noether 1-form as a differential form on the space of boundary fields.

The critical dimension $d=1$ is different from the outset, the Einstein equations are trivial and the group of conformal transformations of the metric has to be considered explicitly. It is interesting to notice that the procedure that we use to construct the boundary structure will be able to single out the critical dimension. The explicit expressions for the boundary data will indeed be singular in $d=1$, hinting at a different behaviour.

This is the first of a series of papers dealing with the BV-BFV approach to quantum gravity. In a subsequent work \cite{CS2} we will show a similar analysis of the Palatini-Holst formalism of GR, highlighting interesting similarities and differences. In \cite{CS3} we will perform the actual boundary BV quantisation of one dimensional examples of pure gravity and cosmological models. There we will show how time evolution can be recovered by means of an apparently silly breaking of diffeomorphism invariance, and how this can be interpreted as a natural operation.

Lastly, we would like to stress that this bulk-boundary correspondence, when extended to the level of ghost fields and antifields, hints at a strong regularity of the theory in terms of compatibility with the gauge symmetries. This is, possibly, to be taken as a necessary requirement for a theory to be regarded as \emph{neatly quantisable}, or at least as a mean of distinction between classical theories, otherwise indistinguishable.

\section{Classical BV and BFV formalisms}\label{Section:background}

We will consider here a general framework for gauge field theories. Consider a \emph{space of fields}, i.e. a (possibly infinite dimensional) $\mathbb{Z}$-graded odd-symplectic manifold $\mathcal{F}$ with a symplectic form $\Omega$ of degree $|\Omega|=k$ together with a local, degree $k+1$ functional $S$ of the fields and a finite number of their derivatives. 

The dynamical content of the theory is encoded in the Euler-Lagrange variational problem for the functional $S$. The $\mathbb{Z}$-grading is sometimes called \emph{ghost number}, but it will be often replaced by the computationally friendly \emph{total degree}, which takes into account the sum of different gradings when the fields belong to some graded vector space themselves (e.g. differential forms).

The symmetries are encoded by an odd vector field $Q\in \Gamma(T[1])F$ such that $[Q,Q]=0$. A vector field with such a property is said to be  \emph{cohomological}.

Among these pieces of data some compatibility conditions are required. We give the following definitions for different values of $k$. According to the convention that we adopt, ordinary symplectic manifolds are called $(0)$-symplectic in the graded setting. Our model for a bulk theory will be given by 
\begin{definition}\label{BVdef}
A BV manifold is the collection of data $(\mathcal{F}, S, Q,  \Omega)$ with $(\mathcal{F},\Omega)$ a $\mathbb{Z}$-graded $(-1)$-symplectic manifold, and $S$ and $Q$ respectively a degree $0$ function and a degree $1$ vector field on $\mathcal{F}$ such that
\begin{enumerate}
\item $\iota_{Q}\Omega=\delta S$, i.e. $S$ is the Hamiltonian function of $Q$
\item $[Q,Q]=0$, i.e. $Q$ is cohomological.
\end{enumerate}
The symplectic structure defines an odd-Poisson bracket $(,)$ on $\mathcal{F}$ and the above conditions together imply
\begin{equation}
(S,S)=0
\end{equation}
the \emph{Classical Master Equation} (CME).
\end{definition}

On the other hand, the model for a boundary theory, \emph{induced} in some sense to be explained, will be given by
\begin{definition}\label{BFVdef}
A BFV manifold is the collection of data $(\mathcal{F}^\partial, S^\partial, Q^\partial,  \omega^\partial)$ with $(\mathcal{F}^\partial,\omega^\partial)$ a $\mathbb{Z}$-graded $0$-symplectic manifold, and $S^\partial$ and $Q^\partial$ respectively a degree $1$ function and a degree $1$ vector field on $\mathcal{F}^\partial$ such that
\begin{enumerate}
\item $\iota_{Q^\partial}\omega^\partial=\delta S^\partial$, i.e. $S^\partial$ is the Hamiltonian function of $Q^\partial$
\item $[Q^\partial,Q^\partial]=0$, i.e. $Q^\partial$ is cohomological.
\end{enumerate}
This implies that $S^\partial$ satisfies the CME. If $\omega^\partial$ is exact, we will say that the BFV manifold is exact.
\end{definition}

These definitions abstract from the following prototype. Usually one starts from a classical theory, that is, for each manifold $M$ of a fixed dimension, the assignment of a local action functional $S_{\text{cl}}$ on some space of classical fields $F_M$ and a distribution in the bulk $D\subset TF_M$ encoding the symmetries, i.e. $L_X(S_{\text{cl}})=0$ for all $X\in \Gamma(D)$. The only requirement on $D$ for the formalism to make sense is that $D$ be involutive on the critical locus of $S_{\text{cl}}$. Notice that $D$ can be the distribution induced by a Lie algebra (group) action, in which case it is involutive on the whole space of fields. When this is the case we will talk of BRST formalism, even though the setting will be slightly different from the original one (for another account on the relationship between the BV and BRST formalism see, e.g. \cite{PavelTh}). We will be mainly interested in these types of theories, but for the sake of completeness we will sketch the general construction.

To construct a BV manifold on a closed manifold $M$ starting from classical data we must first extend the space of fields to accommodate the symmetries: $F_M\leadsto \mathcal{F}_M=T^*[-1]D[1]$. Symmetries are considered with a degree shift of $+1$, whereas the dualisation introduces a different class of fields (called anti-fields) with opposite parity to their conjugate fields, owing to the $-1$ shift in the cotangent functor. This yields a $(-1)$-symplectic manifold, which is a good candidate to be the space of fields we want to work with.

The classical action has to be extended as well to a new local functional on $\mathcal{F}_M$, and if we want this to satisfy the axioms of the BV manifold we must impose the CME on the extended action. This process will a priori need the introduction of higher degree fields to the space of fields in order to resolve, under some regularity assumptions, the relations among degree $1$ fields. This process of extension goes through co-homological perturbation theory \cite{stash97, Stash, BV81, FK, CMR2} and it will ensure us to end up with a BV manifold. However, for a theory which is BRST-like, the extension is determined by the following straightforward result \cite{BV81}:
\begin{theorem}\label{minimalBV}
If $D$ comes from a Lie algebra action, the functional $S_{BV}=S_{\text{cl}} + \langle \Phi^\dag,Q\Phi\rangle$ on the space of fields $\mathcal{F}_M=T^*[-1]D[1]$ satisfies the CME, where $\Phi$ is a multiplet of fields in $D[1]$, $\Phi^\dag$ denotes the corresponding multiplet of conjugate (anti-)fields and $Q$ is the degree $1$ vector field encoding the symmetries of $D$. 

$\mathcal{F}_M$ is then a $(-1)$-symplectic manifold and together with $S_{BV}$ and $Q$ it yields a BV manifold corresponding to a (minimal) extension of the classical theory.
\end{theorem}

\subsection{BV-BFV formalism for gauge theories}

We will explain here in which sense Definition \ref{BFVdef} is a boundary model for Definition \ref{BVdef}. 

\begin{definition}\label{BV-BFV}
An exact BV-BFV pair is the data given by an exact BFV manifold $(\mathcal{F}^\partial,\omega^\partial=\delta \alpha^\partial, S^\partial, Q^\partial)$ together with a map $\pi\colon \mathcal{F}\longrightarrow \mathcal{F}^\partial$ such that 
\begin{enumerate}
\item The map ${\pi}$ is a surjective submersion, with $\Omega=\pi^*\omega^\partial$.
\item There is a vector field $Q\in\mathfrak{X}[1](\mathcal{F})$ such that $Q^\partial = \pi_*Q$ and $[Q,Q]=0$ is cohomological.
\item The BV-BFV formula $\iota_Q\Omega=\delta S + {\pi}^*{\alpha}^\partial$ is satisfied.
\end{enumerate}
Such a pair will be denoted by $(\mathcal{F},\mathcal{F}^\partial)_{\pi}$
\end{definition}

In field theory there are natural examples of BV-BFV pairs as in the following prototypical construction. Say that we start from the data defining a BV manifold, but this time we allow $M$ to have a boundary $\partial M$: the requirement that $\iota_{Q}\Omega=\delta S$ is (in general) no longer true. What will happen is that the integration by parts one usually has to take into account when computing $\delta S$ will leave some non zero terms \emph{on the boundary}. More precisely, consider the map 
\begin{equation}
\widetilde{\pi}_M:\mathcal{F}_M \longrightarrow \widetilde{\mathcal{F}}_{\partial M}
\end{equation}
that takes all fields and their jets to their restrictions to the boundary  (it is a surjective submersion). We can interpret the boundary terms as the pullback of a one form.  $\widetilde{\alpha}$ on $\widetilde{\mathcal{F}}$, namely
\begin{equation}
\iota_{Q}\Omega=\delta S + \widetilde{\pi}_M^*\widetilde{\alpha}
\end{equation}
We will call $\widetilde{\alpha}$ the \emph{pre-boundary} one form. In full generality $\widetilde{\alpha}$ is a connection on a line bundle, yet when $S$ is a function on the space of fields, $\widetilde{\alpha}$ is a globally well defined $1$-form.

Notice that if we are given this data, we can interpret this as a \emph{broken BV manifold}, with some relation to the boundary. We can in fact consider the pre-boundary two form $\widetilde{\omega}\coloneqq \delta\widetilde{\alpha}$ and if it is pre-symplectic (i.e. its kernel has constant rank) then we can define the true space of boundary fields $\mathcal{F}^\partial_{\partial M}$ to be the symplectic reduction of the space of pre-boundary fields, namely:
\begin{equation}\label{smoothred}
\mathcal{F}_{\partial M}^\partial=\qsp{\widetilde{\mathcal{F}}_{\partial M}}{\mathrm{ker}(\widetilde{\omega})}
\end{equation}
with projection to the quotient denoted by $\varpi:\widetilde{\mathcal{F}}_{\partial M}\longrightarrow \mathcal{F}_{\partial M}^\partial$. If all of the above assumptions are satisfied, the map $\pi_M\coloneqq \varpi\circ\widetilde{\pi}$ is a surjective submersion, the reduced two form $\omega^\partial\coloneqq \underline{\widetilde{\omega}}$ is a $0$-symplectic form, and the key result is

\begin{proposition}[CMR \cite{CMR2}]\label{CMRprop}
The cohomological vector field $Q$ projects to a cohomological vector field $Q^\partial$ on the space of boundary fields $\mathcal{F}_{\partial M}^\partial$. Moreover $Q^\partial$ is Hamiltonian for a function $S^\partial$, the boundary action.
\end{proposition}

When this construction goes through, it associates to a manifold with boundary $(M,\partial M)$ a BV-BFV pair that depends on the manifold data. We will say that

\begin{definition}
A d-dimensional BV-BFV theory is an association of a BV-BFV pair $(\mathcal{F}_M,\mathcal{F}_{\partial M}^\partial)_{\pi_M}$ to a d-dimensional manifold with boundary $(M,\partial M)$.
\end{definition}

To summarise the construction and rephrase Proposition \ref{CMRprop} we have the following

\begin{theorem}[CMR]
Whenever the space $\mathcal{F}^\partial_{\partial M}$ of Equation \eqref{smoothred} is smooth, we are given the BV-BFV pair $(\mathcal{F}_M,\mathcal{F}_{\partial M}^\partial)_{\pi_M}$. The construction of a BV manifold for a local field theory on a closed manifold $M$ extends to a (possibly exact) BV-BFV theory on the manifold with boundary $(M,\partial M)$.
\end{theorem}

\begin{remark}
Notice that in Definition \ref{BFVdef} it is possible to relax the requirement that the BFV 2-form $\widetilde{\omega}$ be nondegenerate and introduce the notion of \emph{pre-BFV} manifolds. When that is the case we may define \emph{pre-BV-BFV} pairs to be modeled over these more general pre-BFV manifolds. Observe that a pre-BV-BFV pair $(\mathcal{F},\widetilde{\mathcal{F}})_{\widetilde\pi}$ such that the form $\mathrm{ker}(\widetilde{\omega})$ is a subbundle gives naturally rise to a BV-BFV pair on the symplectic reduction $\mathcal{F}^\partial = \widetilde{\mathcal{F}}\big\slash\mathrm{ker}(\widetilde{\omega})$, by composing $\widetilde{\pi}$ with the symplectic reduction map.
\end{remark}

Some compatibility between bulk and boundary can always be achieved in terms of the space of pre-boundary fields, on which the differential of the Noether $1$-form is degenerate. The crucial assumption is that the symplectic reduction of this $2$-form should be smooth. In this paper we show that this is satisfied in the case of the Einstein-Hilbert formulation of General Relativity.

The advantage of such a point of view is at least twofold. First of all, as we just saw, the formalism is \emph{large enough} to be able to describe consistently what happens both in the bulk and in the boundary. On the other hand it is flexible enough to allow for symmetries that are more general than a Lie group action. For instance it is possible to accomodate symmetries that close only on shell (e.g. Poisson sigma model) or symmetries whose generators are not linearly independent, where higher relations among the relations are required (e.g. BF theory or other theories involving $(d>1)$-differential forms.

The BV manifold that we have constructed in Theorem \ref{minimalBV} when a gauge theory of the BRST-kind was given is sometimes called the \emph{minimal BV extension} of the gauge theory. When a non trivial boundary is allowed, we will use this minimal extension as the starting point for the BV-BFV analysis.

What one aims to establish is whether this minimal BV extension on the bulk does indeed yield a BV-BFV theory. We will accomplish this for the special case of General Relativity in the EH-ADM formalism  in what follows.

\section{General Relativity in the BV formalism}\label{Section:ADM}
General Relativity (GR) is a classical theory modelling the gravitational interaction between physical objects. It can be expressed as the variational problem for a functional on some space of classical fields. The usual presentation of GR in the second order formalism requires one to consider the basic field to be a pseudo-Riemannian metric $g$ on a \emph{spacetime} manifold $M$, i.e. a $d+1$ dimensional manifold, possibly with non trivial boundary $\partial M$, and the Einstein-Hilbert action (up to constants)
\begin{equation}
S^{EH}_{cl}=\int\limits_{M}(R[g] - 2\Lambda)\,\sqrt{-\mathrm{g}}\,d^{d+1}x
\end{equation}
where $R[g]$ is the Ricci scalar of the pseudo-Riemannian metric $g\in\mathcal{PR}_{(d,1)}(M)$ with signature $(d,1)$, $\mathrm{g}\coloneqq   \mathrm{det}(g)$ and $\Lambda$ is the cosmological constant. The Einstein equations, which determine the dynamics of $g$, are derived from the action as Euler-Lagrange equations for the variational problem.

The very principle of general covariance, from which the whole of General Relativity stems, states that the theory should be the same in all reference frames or, equivalently, that the equations of motion should not depend on the choice of coordinates on $M$, and this is mathematically encoded in the action being invariant under the full group of diffeomorphism of the manifold $M$.

To make contact with the general theory outlined in Section \ref{Section:background} we will consider the distribution defining the (infinitesimal) gauge symmetries $\mathcal{D}$ to be the distribution on the space of fields given by all vector fields on $M$, acting on the fields \emph{via} Lie derivative. This will be encoded in the BV formalism by the introduction of a new field $\xi\in \Gamma\left(T[1]M\right)$, which is an odd vector field representing such infinitesimal symmetries, and the space of fields will be enlarged accordingly, as we shall see in detail later on.

It turns out that a clean result can be obtained, when some slightly restricting conditions are imposed on the boundary $\partial M$: namely, we will require that $\partial M$ be either space-like or time-like, and this will allow us a clever rewriting of the action in a \emph{boundary-friendly} fashion. 

To achieve this in field theory, since the pseudo-Riemannian structure is not fixed, we will need to restrict the space of fields to only those metrics whose restriction to the boundary has either space-like or time-like signature. Such a space of fields will be denoted by $\mathcal{PR}^{\partial M}_{(d,1)}(M)$.

Observe that in the literature (e.g. \cite{ADM,dWi}), it is customary to require that the spacetime manifold $M$ be \emph{globally hyperbolic}, or that it has the product structure $\Sigma \times \RR$ for $\Sigma$ an embedded space/time-like submanifold of $M$. This would indeed be a much stronger requirement, and in fact we only require it be true in a neighborhood of the boundary. Notice that any globally hyperbolic Lorentzian structure for the space-time $M$ is in particular contained in $\mathcal{PR}^{\partial M}_{(d,1)}(M)$.

\subsection{ADM decomposition and boundary structure}
Among the possible boundaries that one can consider there is the class of non null boundary, described locally by a submanifold of the form $x^n=\text{const}$. As we mentioned, this is equivalent to (or rather means) asking that the space of pseudo-Riemannian structures on the manifold with boundary $M$ be limited to those metrics whose restriction to the boundary has either time-like or space-like signature.

When this is the case, and when the $x^n$ component corresponds to a signature $-\epsilon$, the metric and its inverse are re-written in the form:
\begin{equation}\label{ADMmetric}
\begin{array}{c}
g_{\mu\nu}=\epsilon\left(\begin{array}{cc} -(\eta^2 - \beta_a\beta^a) & \beta_b \\ \beta_a & \gamma_{ab}\end{array}\right)\\
\\
g^{\mu\nu}=\epsilon\eta^{-2}\left(\begin{array}{cc} -1 & \beta^b \\ \beta^a & \eta^2\gamma^{ab}-\beta^a\beta^b\end{array}\right)\end{array}
\end{equation}
where $\eta$ and $\beta_a$ are functions for all $a=1,2,3$. For  simplicity $\epsilon=1$ if $x^n$ is a timelike direction, that is to say when the boundary is spacelike.

With this decomposition we have that $\sqrt{-\mathrm{g}}=\eta\sqrt{|\gamma|_\epsilon}$, with $\gamma = \mathrm{det}\gamma_{ij}$, and $|\gamma|_\epsilon$ means that we shall consider the absolute value of the determinant if needed (if $\epsilon=-1$). We will understand this fact from now on and simply write $\sqrt\gamma$. The classical Einstein-Hilbert action gets rewritten as
\begin{equation}\begin{aligned}\label{fullADMaction}
S=\int\limits_M&\Big\{\eta\sqrt{\gamma}(\epsilon(K_{ab}K^{ab} - K^2) +R^\partial -2\Lambda)- 2\partial_n(\sqrt{\gamma}K) +2\partial_a(\sqrt{\gamma} K \beta^a - \sqrt{\gamma}\gamma^{ab}\partial_b\eta)\Big\}d^{d+1}x
\end{aligned}\end{equation}
where we define $K_{ab}$,  the second fundamental form of the boundary submanifold, and its trace $K$ by means of the boundary covariant derivative $\nabla^\partial$ as follows
\begin{align}\label{Kappa e Ti}
K_{ab}&=\frac12 \eta^{-1}(2\nabla^\partial_{(a}\beta_{b)} - \partial_n\gamma_{ab})=\frac12\eta^{-1}T_{ab}\\
K&=\gamma^{ab}K_{ab}=\frac12\eta^{-1}\gamma^{ab}T_{ab}=\frac12\eta^{-1}T
\end{align}
while $T_{ab}$ and $T$ are introduced for later convenience.
We will redefine the ADM Lagrangian as
\begin{equation}
L_{ADM}\coloneqq \eta\sqrt{\gamma}(\epsilon(K_{ab}K^{ab} - K^2) +R^\partial -2\Lambda)
\end{equation}
The classical space of fields in this case is then simply given by $\mathcal{F}_{\text{cl}}=\mathcal{PR}^{\partial M}_{(d,1)}(M)$ the space of pseudo Riemannian metrics on $M$ with signaure (d,1), and space/time-like signature when restricted to the boundary.

\begin{remark}
The total normal derivative appearing in \eqref{fullADMaction} is the so called \emph{Gibbons-Hawking-York} boundary term. In our framework, it will only affect the boundary $1$-form by an exact term, and will not interfere with the rest of the boundary structure. The action we will consider from now on is 
\begin{equation}\label{ADMaction}
S_{ADM}=\int_M L_{ADM}
\end{equation}
\end{remark}

\subsection{Classical boundary structure}\label{Sect:ClassADM}
To start off, we will consider first the classical (i.e. non-BV) structure that is induced on the boundary. This is often called \emph{canonical analysis}, and one replaces the Lagrangian description with the Hamiltonian in the phase space of the system. The advantage in applying our variational approach to the classical case as well, is that we are able to perform the symplectic reduction of the space of classical pre-boundary fields, to find a well defined symplectic structure on the space of classical boundary fields, i.e. the phase space, encoding the canonical relations in a straightforward way.

\begin{proposition}\label{Prop:ClAdm}
The space of classical boundary fields for General Relativity in the ADM formalism for any dimension $d+1\not=2$ is an exact symplectic manifold. In a local chart the symplectic form reads
\begin{equation}
\omega^\partial=\epsilon \int\limits_{\partial M} \delta\boldsymbol{\gamma}^{ab}  \delta\boldsymbol{\Pi}_{ab}
\end{equation}
where the projection map to the boundary fields reads:
\begin{equation}
\pi_M\colon\begin{cases}
\boldsymbol{\gamma}_{ij}=\gamma_{ij}\\
\boldsymbol{\Pi}_{lm}=\frac{\sqrt{\gamma}}{2}\left(\tJ{lm} - \gamma_{lm}\gamma^{ij}\tJ{ij}\right)
\end{cases}
\end{equation}
with
\begin{equation}
\tJ{lm}=\eta^{-1}\left(J_{lm} - 2\nabla_{(l}\beta_{m)}\right)
\end{equation}
and $J_{ab}$ is the first normal jet of $\gamma_{ab}$ evaluated at the boundary.
\end{proposition}
\begin{proof}
Consider the variation of the action $S_{ADM}$: it splits in a bulk term and a boundary term. The latter is interpreted as a $1$-form $\widetilde{\alpha}$ on the space of pre-boundary fields $\widetilde{\mathcal{F}}^{\text{cl}}$, which is given by restrictions of the bulk metric and its normal jets $J_{ab}\coloneqq \partial_n\gamma_{ab}\big|_{\partial M}$ to $\partial M$: 
\begin{equation}\label{ADM1form}
\widetilde{\alpha}=2\epsilon\int\limits_{\partial M}\left\{\delta(\sqrt{\gamma}\gamma^{ab})K_{ab} - \frac{\sqrt{\gamma}}{2}\delta\gamma^{ab}K_{ab}\right\}
\end{equation}
and the two form $\widetilde{\omega}=\delta\widetilde{\alpha}$, using the definitions \eqref{Kappa e Ti} of $K_{ab}$ and ${T_{ab}}$ is
\begin{multline}
\widetilde{\omega}=\epsilon\int\limits_{\partial M}\Bigg\{ \delta\eta^{-1}\delta(\sqrt{\gamma}\gamma^{ab})T_{ab} - \eta^{-1}\delta(\sqrt{\gamma}\gamma^{ab})\delta T_{ab} - \eta^{-1}\frac{\delta\sqrt{\gamma}\delta\gamma^{ab}}{2}T_{ab} - \delta\eta^{-1}\frac{\sqrt{\gamma}}{2}\delta\gamma^{ab}T_{ab} + \eta^{-1}\frac{\sqrt{\gamma}}{2}\delta\gamma^{ab}\delta T_{ab}\Bigg\}
\end{multline}
Both $K$ and $T$ are functions of $g$ and $J$. Observe that the transversal jets $J_{na}$ are not present because of the clever rewriting of the action.


After some straightforward calculations one is able to gather that the kernel of the two form is given, for $d\not=1$ by 
\begin{align}
(X_\gamma)^{ab}&=0 \label{Xgamma}\\
(X_T)_{lm}&=- \eta(X_{\eta^{-1}})T_{lm} \label{XKappa}
\end{align}
but the $(X_{\beta_m})$ component of a vector field in the kernel turns out to be free, as well as the $\eta^{-1}$ component. In fact, equation \eqref{XKappa} can be unfolded to yield:
\begin{equation}
(X_J)_{lm}= - \eta(X_{\eta^{-1}})J_{lm} + 2\nabla_{(l}(X_\beta)_{m)} + 2\eta(X_{\eta^{-1}})\nabla_{(l}\beta_{m)}
\end{equation}

The generators in the kernel are 
\begin{subequations}\label{kernelCLADM}\begin{align}
\mathbb{E}^{-1}&=(X_{\eta^{-1}})\pard{}{\eta^{-1}} - \eta(X_{\eta^{-1}})J_{lm} \pard{}{J_{lm}} + 2\eta(X_{\eta^{-1}})\nabla_{(l}\beta_{m)}\pard{}{J_{lm}}\\
\mathbb{B}_l&= (X_\beta)_l\pard{}{\beta_l} + 2\nabla_{(l}(X_\beta)_{m)}\pard{}{J_{lm}}
\end{align}\end{subequations}
and thus, solving the differential equations given by the kernel vector fields \eqref{kernelCLADM} together with Equation \eqref{Xgamma} yields the symplectic reduction map:
\begin{equation}
\pi\colon\begin{cases}
\widetilde{\gamma}_{ij}=\gamma_{ij}\\
\widetilde{J}_{lm}=\eta^{-1}\left(J_{lm} - 2\nabla_{(l}\beta_{m)}\right)
\end{cases}
\end{equation}
It is a matter of a simple check to verify that the one form
\begin{equation}
\alpha^\partial=\epsilon\int\limits_{\partial M}\frac{\sqrt{\widetilde{\gamma}}}{2}\left(\delta\widetilde{\gamma}^{ij}\widetilde{\gamma}_{ij}\widetilde{\gamma}^{lm}\widetilde{J}_{lm}  - \delta\widetilde{\gamma}^{lm}\widetilde{J}_{lm} \right)
\end{equation}
pulls back along $\pi$ to the pre-boundary one form $\widetilde{\alpha}$: 
\[
\pi^*\alpha^\partial = \widetilde{\alpha}
\]
which is horizontal, i.e. $\iota_{\mathbb{E}^{-1}}\widetilde{\alpha}=\iota_{\mathbb{B}_l}\widetilde{\alpha}=0$. This implies that the symplectic manifold $\left(\mathcal{F}^\partial,\delta \alpha^\partial\right)$ is exact. $\pi_M$ is then obtained composing $\pi$ with the restriction to the pre-boundary fields $\widetilde{\pi}$.

Introducing the new variables $\boldsymbol{\gamma}^{ab}\equiv\widetilde{\gamma}^{ab}$ and $\boldsymbol{\Pi}_{lm}=\frac{\sqrt{\gamma}}{2}\left(\tJ{lm} - \gamma_{lm}\gamma^{ij}\tJ{ij}\right)$ we have
\begin{equation}
\alpha^\partial=-\epsilon\int\limits_{\partial M} \delta\boldsymbol{\gamma}^{ab} \boldsymbol{\Pi}_{ab} \Longrightarrow \omega^\partial=\epsilon \int\limits_{\partial M} \delta\boldsymbol{\gamma}^{ab}  \delta\boldsymbol{\Pi}_{ab}
\end{equation}
which is the symplectic form in the space of classical boundary fields, expressed in local Darboux coordinates. 
\end{proof}

\begin{remark}
We managed to recover the phase space description of General Relativity in the symplectic framework. Notice that in the non-BV setting the compatibility with the boundary structure is encoded in the boundary term $\pi_M^*\alpha^\partial_{\partial M}$, a failure of the variation of the action from being given by the Euler-Lagrange equations alone. When turning to the BV theory we will see how this compatibility can be enriched to yield the full fundamental formula \eqref{QADM}.
\end{remark}

\begin{remark}
Observe that we have performed a symplectic reduction that encodes the usual canonical analysis of General Relativity (in the ADM formalism). Our boundary field $\boldsymbol{\Pi}_{ab}$ is a projected version of the usual (i.e. literature) \emph{momentum} coordinate conjugate to $\gamma^{ab}$ (let us call it $p_{ab}=\pi^*\boldsymbol{\Pi}_{ab}$), with the difference that in the present case the \emph{conjugacy} is in the symplectic sense, as we quotient by the kernel of the pre-symplectic form $\widetilde{\omega}$.
\end{remark}

In what follows we will show how this can be extended to the BV setting, which explicitly encodes the symmetries. This will allow us to recover the usual energy and momentum constraints in a straightforward way, still holding on to the clean symplectic description of the phase space.

\subsection{BV-BFV ADM theory}
Recalling the general theory we outlined in Section \ref{Section:background}, in order to perform a consistent analysis of the theory including the symmetries, one has to find the correct BV data. The geometric information we need is the distribution in the space of fields that generates the symmetries.

In our case, General Relativity is invariant under the action of the whole diffeomorphism group of the space-time manifold $M$. The theory can be treated as a BRST-like theory since the symmetry algebra $\Gamma(TM)$ closes everywhere in the space of fields, and we can use Theorem \ref{minimalBV} to extend the classical ADM action to its BV-extended counterpart. Indeed we consider the following action:
\begin{equation}
S_{ADM}^{BV}=S_{ADM} - \int\limits_M\left(L_\xi g\right)g^\dag + \int\limits_M\frac{1}{2}\iota_{[\xi,\xi]}\xi^\dag\equiv S_{ADM} + S_{BV}
\end{equation}
with $S_{ADM}$ as in \eqref{ADMaction}, since the action of the cohomological vector field $Q$ reads
\begin{equation}
\begin{aligned}
Qg&=L_\xi g\\
Q\xi&=\frac{1}{2}[\xi,\xi]
\end{aligned}
\end{equation} 
and $\xi\in \Gamma(T[1]M)$ a generic vector field, declared to have ghost number $1$. The space of fields is then given by the shifted cotangent bundle: 
\begin{equation}
\mathcal{F}_{ADM}\coloneqq T^*[-1]\left[\mathcal{PR}^{\partial M}_{d+1}(M)\oplus \Gamma\left(T[1]M\right)\right].
\end{equation}
equipped with the canonical odd-symplectic form $\Omega_{BV}$. Our first result in this setting is the following

\begin{theorem}\label{Theo:ADM1}
For all $d\not=1$, the data $(\mathcal{F}_{ADM}, S_{ADM}^{BV},Q,\Omega_{BV})$ induce an exact BV-BFV theory. The induced data on the boundary will be denoted by $(\mathcal{F}^\partial,S^\partial,Q^\partial, \omega^\partial)$. In particular we have that $Q^\partial=\pi_M{}_* Q$, and $\iota_{Q^\partial}\omega^\partial = \delta S^\partial$.
\end{theorem}
\begin{proof}
The variation of $S_{ADM}^{BV}$ induces the following pre-boundary $1$-form on the space of pre-boundary fields $\widetilde{\mathcal{F}}_{ADM}$, which is given by restrictions of the bulk fields  to $\partial M$, together with the normal jets of the boundary metric $J_{ab}\coloneqq \partial_n\gamma_{ab}\big|_{\partial M}$:
\begin{multline}\label{BVADMoneform}
\widetilde{\alpha}_{ADM}=2\epsilon\int\limits_{\partial M}\eta\left\{\delta(\sqrt{\gamma}\gamma^{ab})K_{ab} - \frac{\sqrt{\gamma}}{2}\delta\gamma^{ab}K_{ab}\right\}\\
+2\epsilon\int\limits_{\partial M}\left((-\eta^2 + \beta_a\beta^a)\delta\xi^ng^\dag{}^{nn} + \beta_a\delta\xi^ng^\dag{}^{an} + \beta_a\delta\xi^ag^\dag{}^{nn} + \gamma_{ab}\delta\xi^{(a}g^\dag{}^{b)n}\right)\\
-\epsilon \int\limits_{\partial M}\left(\xi^n\delta(-\eta^2 + \beta_a\beta^a)g^\dag{}^{nn} +2\xi^n\delta \beta_a g^\dag{}^{an} + \xi^n\delta\gamma_{ab}g^\dag{}^{ab}\right)  - \int\limits_{\partial M}\xi^n\delta\xi^\rho\chi_\rho
\end{multline}
and $2$-form $\widetilde\omega=\delta\widetilde{\alpha}_{ADM}$:
\begin{multline}
\widetilde\omega=\int\limits_{\partial M}\epsilon\delta\left\{\eta^{-1}\delta(\sqrt{\gamma}\gamma^{ab})K_{ab} -\eta^{-1} \frac{\sqrt{\gamma}}{2}\delta\gamma^{ab}K_{ab}\right\}  -  \delta\xi^n\delta\xi^\rho\chi_\rho  + \xi^n\delta\xi^\rho\delta\chi_\rho \\
+ \epsilon\Big(\delta(-\eta^2 + \beta_a\beta^a)\delta\xi^ng^\dag{}^{nn} +2(-\eta^2+ \beta_a\beta^a)\delta\xi^n\delta g^\dag{}^{nn}  + 2\beta_a\delta\xi^n\delta g^\dag{}^{an}\\
+ 2\delta\beta_a\delta\xi^ag^\dag{}^{nn} + 2\beta_a\delta\xi^a \delta g^\dag{}^{nn} + 2\delta\gamma_{ab}\delta\xi^{(a} g^\dag{}^{b)n} + 2\gamma_{ab}\delta\xi^{(a}\delta g^\dag{}^{b)n}\Big) \\
-\epsilon\Big(\xi^n\delta(-\eta^2 + \beta_a\beta^a)\delta g^\dag{}^{nn} +2\xi^n\delta\beta_a\delta g^\dag{}^{an} + \delta\xi^n\delta\gamma_{ab}g^\dag{}^{ab} + \xi^n\delta\gamma_{ab}\delta g^\dag{}^{ab}\Big)
\end{multline}
where we fixed the volume form $v=dx^n\wedge v^\partial$. 

More explicitly, the space of pre-boundary fields in degree zero is given by $3$-metrics on the boundary $\gamma$ and their first normal jets $J$ times (d+1) copies of functions on the boundary representing the components $(\eta,\beta_a)$ of the metric $g_{\mu\nu}$ restricted to the boundary. In degree one it contains vector fields tangent to the boundary and the function of degree one $\xi^n$, which is the normal component of the vector field $\xi$ restricted to the boundary. Finally it contains the symmetric covariant 2-tensors $g^\dag{}^{\mu\nu}$ restricted to the boundary in degree $-1$ and functions $\chi_\rho$ of degree $-2$ on the boundary, coming from the components of the restriction to the boundary of the anti-ghost $\chi$. With an abuse of notation we will denote the restrictions of the fields to the boundary with the same symbol.

Recalling that $K_{ab}$ is a function of $J_{ab}\coloneqq  \partial_n\gamma_{ab}|_{\partial M}$, it is just a matter of lengthy computations to show that $\widetilde{\omega}$ is presymplectic: indeed,  excluding the case $d=1$, the equations defining the kernel can be solved to yield:
\begin{subequations}\begin{align}
(X_J)_{lm}=&+\eta^{-1}(X_\eta) J_{lm} + 2\nabla_{(l}(X_\beta)_{m)} - 2\eta^{-1}(X_\eta)\nabla_{(l}\beta_{m)} + \frac{4}{\sqrt{\gamma}}(X_\eta)\left(\frac{1}{d-1}\gamma_{lm}\beta_a - \beta_{(l}\gamma_{m)a}\right)g^\dag{}^{an}\xi^n \\\notag
&-\frac{4}{\sqrt{\gamma}}\eta\left(\frac{1}{d-1}\gamma_{lm}(X_\beta)_a - (X_\beta)_{(l}\gamma_{m)a}\right)g^\dag{}^{an}\xi^n +\frac{2}{\sqrt{\gamma}}(X_\eta)\left(\frac{1}{d-1}\gamma_{lm}\gamma_{ab} - \gamma_{la}\gamma_{bm}\right)g^\dag{}^{ab}\xi^n\\\notag
&-\frac{2}{\sqrt{\gamma}}\eta\left(\frac{1}{d-1}\gamma_{lm}\gamma_{ab} - \gamma_{al}\gamma_{bm}\right)(X_{g^\dag})^{ab}\xi^n\\
(X_g^\dag)^{bn}=&-\gamma^{ab}(X_\beta)_ag^\dag{}^{nn} +  \eta^{-1}\beta^b(X_\eta)g^\dag{}^{nn}  + \epsilon\eta^{-3}\beta^b\beta^a\chi_a(X_\eta)\xi^n - \frac\epsilon2\eta^{-2}\beta^b\beta^a(X_\chi)_a\xi^n   \\\notag
&- \frac\epsilon2 \eta^{-2}\beta^b(X_\beta)_c\gamma^{cd}\chi_d\xi^n  - \epsilon\eta^{-3}\beta^b(X_\eta)\chi_n\xi^n  +\frac\epsilon2\eta^{-2}\beta^b(X_\chi)_n\xi^n  - \frac\epsilon2 \eta^{-1}\gamma^{ba}\chi_a(X_\eta)\xi^n  + \frac\epsilon2\gamma^{ba}(X_\chi)_a\xi^n\\
(X_{g^\dag})^{nn}=&-\eta^{-1}(X_\eta)g^\dag{}^{nn}  - \epsilon\eta^{-3}(X_\eta)\beta^a\chi_a\xi^n  + \frac\epsilon2 \eta^{-2} \beta^a(X_\chi)_a\xi^n  + \frac\epsilon2\eta^{-2}(X_\beta)_b\gamma^{ab}\chi_a\xi^n   \\\notag
&+ \epsilon\eta^{-3}(X_\eta)\chi_n\xi^n  - \frac\epsilon2\eta^{-2}(X_\chi)_n\xi^n\\
(X_\xi)^a = &+\beta^a\eta^{-1}(X_\eta)\xi^n - \gamma^{ab}(X_\beta)_b\xi^n\\
(X_\xi)^n =& -\eta^{-1}(X_\eta)\xi^n 
\end{align}\end{subequations}

and the kernel is generated by the (vertical) vector fields:

\begin{subequations}\label{kernelBVADM}\begin{align}
\mathbb{X}_{(n)}=&(X_\chi)_n\pard{}{\chi_n}  - \frac\epsilon2\eta^{-2}(X_\chi)_n\xi^n\pard{}{g^\dag{}^{nn}}  + \frac\epsilon2\beta^b\eta^{-2}(X_\chi)_n\xi^n\pard{}{g^\dag{}^{bn}} \\ 
\mathbb{X}_{(a)} =& (X_\chi)_a\pard{}{\chi_a}  + \frac\epsilon2 \eta^{-2} \beta^a(X_\chi)_a\xi^n\pard{}{g^\dag{}^{nn}}  -\left(\frac\epsilon2\eta^{-2}\beta^b\beta^a(X_\chi)_a\xi^n - \frac\epsilon2\gamma^{ba}(X_\chi)_a\xi^n\right)\pard{}{g^\dag{}^{bn}}\\\label{Betavert}
\mathbb{B}_{(a)}=&(X_\beta)_a\pard{}{\beta_a}  - \gamma^{ab}(X_\beta)_a\xi^n \pard{}{\xi^b}  + \frac\epsilon2\eta^{-2}\gamma^{ab}(X_\beta)_a\chi_b\xi^n\pard{}{g^\dag{}^{nn}} \\\notag
&+\left(2\nabla_{(l}(X_\beta)_{m)} + \frac{4\epsilon}{\sqrt{\gamma}}\eta\left((X_\beta)_{(l}\gamma_{m)a} - \frac{1}{d-1}\gamma_{lm}(X_\beta)_a \right)g^\dag{}^{an}\xi^n \right)\pard{}{J_{lm}}\\\notag
& + \left( -\frac\epsilon2\eta^{-2} \beta^b\gamma^{cd}(X_\beta)_c\chi_d\xi^n -\gamma^{ab}(X_\beta)_ag^\dag{}^{nn}\right)\pard{}{g^\dag{}^{bn}}\\
\mathbb{G}^\dag{}^{(ab)} =& (X_{g^\dag})^{ab}\pard{}{g^\dag{}^{ab}} +\frac{2\epsilon}{\sqrt{\gamma}}\eta\left( \gamma_{al}\gamma_{bm} - \frac{1}{d-1}\gamma_{lm}\gamma_{ab}\right)(X_{g^\dag})^{ab}\xi^n \pard{}{J_{lm}}\\
\mathbb{E}=&(X_\eta)\pard{}{\eta}  -\eta^{-1}(X_\eta)\xi^n \pard{}{\xi^n} + \beta^a\eta^{-1}(X_\eta)\xi^n\pard{}{\xi^a}  \\\notag
& - \eta^{-1}(X_\eta)g^\dag{}^{nn}\pard{}{g^\dag{}^{nn}} - \epsilon\eta^{-3}\left(\beta^a\chi_a - \chi_n \right)(X_\eta)\xi^n \pard{}{g^\dag{}^{nn}} \\\notag
&  -\left( \epsilon\eta^{-3}\beta^b\chi_n - \epsilon\eta^{-3}\beta^b\beta^a\chi_a  + \frac\epsilon2 \eta^{-1}\gamma^{ba}\chi_a\right)(X_\eta)\xi^n\pard{}{g^\dag{}^{bn}}\\\notag
&- \frac{4\epsilon}{\sqrt{\gamma}}(X_\eta)\left(\beta_{(l}\gamma_{m)a} - \frac{1}{d-1}\gamma_{lm}\beta_a \right)g^\dag{}^{an}\xi^n\pard{}{J_{lm}}\\\notag
&-\frac{2\epsilon}{\sqrt{\gamma}}(X_\eta)\left( \gamma_{la}\gamma_{bm} - \frac{1}{d-1}\gamma_{lm}\gamma_{ab}\right)g^\dag{}^{ab}\xi^n\pard{}{J_{ab}}\\\notag
& +  \eta^{-1}\beta^a(X_\eta)g^\dag{}^{nn}\pard{}{g^\dag{}^{an}}+\eta^{-1}(X_\eta)\left(J_{lm} - 2\nabla_{(l}\beta_{m)}\right)\pard{}{J_{lm}}
\end{align}\end{subequations}
It is easy to check that the boundary one form \eqref{BVADMoneform} is annihilated by all vertical vector fields \eqref{kernelBVADM}, and it is therefore basic, proving the exactness of the BV-BFV pair and concluding the proof.
\end{proof}

Not only can we prove the existence of  a well defined exact BFV structure on the boundary $\partial M$, but it is possible to express it in Darboux coordinates. The explicit expression in a local chart is established by the following

\begin{theorem}\label{Theo:ADMBSTR}
The surjective submersion $\pi_M:\mathcal{F}_{ADM}\longrightarrow \mathcal{F}_{ADM}^\partial$ is given by the local expression:
\begin{equation}\label{projADM}
\pi_M\colon\begin{cases}
\boldsymbol{\Pi}_{lm}&=\frac{\sqrt{\tgam}}{2}\left(\tJ{lm} - \tgam_{lm}\tgam^{ij}\tJ{ij}\right)\\
\boldsymbol{\varphi}_n&=-2\left\{\eta g^\dag{}^{nn}  - \frac{\epsilon}{2}\eta^{-1}\left( \beta^a\chi_a- \chi_n \right)\xi^n \right\}\\
\boldsymbol{\varphi}_a&=2\,\gamma_{ab}\left\{g^\dag{}^{bn} + \gamma^{ba}\beta_a g^\dag{}^{nn}  -  \frac\epsilon2\gamma^{ba}\chi_a\xi^n \right\}\\
\boldsymbol{\xi}^{b}&=\xi^b + \gamma^{ba}\beta_a\xi^n \\
\boldsymbol{\xi}^{n}&=\eta\,\xi^n\\
\boldsymbol{\gamma}_{ab}&=\gamma_{ab}
\end{cases}
\end{equation}
with
\begin{align*}
\tJ{lm}&=\Bigg\{\eta^{-1}\left(J_{lm} - 2\nabla_{(l}\beta_{m)}\right) - \frac{2\epsilon}{\sqrt\gamma}\left(\gamma_{al}\gamma_{bm} - \frac{1}{d-1}\gamma_{lm}\gamma_{ab}\right)g^\dag{}^{ab}\xi^n\\
&-\frac{4}{\sqrt\gamma}\,\epsilon\left(\beta_{(l}\gamma_{m)b}-\frac{1}{d-1}\gamma_{lm}\beta_b\right)g^\dag{}^{bn}\xi^n - \frac{2\epsilon}{\sqrt\gamma}\left(\beta_{(l}\beta_{m)} - \frac{1}{d-1}\gamma_{lm}\beta_b\beta^b\right)g^\dag{}^{nn}\xi^n\Bigg\}
\end{align*}
The boundary symplectic structure on the space of boundary fields in these coordinates ($\rho=\{n,a\}$) reads:
\begin{equation}
\omega^\partial=\epsilon\int\limits_{\partial M} \delta \boldsymbol{\gamma}^{ab}\delta\boldsymbol{\Pi}_{ab} + \delta\boldsymbol{\xi}^\rho\delta\boldsymbol{\varphi}_{\rho}.
\end{equation}
Moreover, the boundary action is given by the expression
\begin{align}\notag
S^\partial=&\int\limits_{\partial M}\Bigg\{\frac{\epsilon}{\sqrt{\boldsymbol{\gamma}}}\left(\boldsymbol{\Pi}^{ab}\boldsymbol{\Pi}_{ab} - \frac{1}{d-1}\boldsymbol{\Pi}^2\right) + \sqrt{\boldsymbol{\gamma}}\left(R^\partial -2\Lambda\right) + \epsilon\partial_a\left(\boldsymbol{\xi}^{a}\boldsymbol{\varphi}_n\right) - \epsilon\boldsymbol{\gamma}^{ab}\boldsymbol{\varphi}_b\partial_a\boldsymbol{\xi}^{n}\Bigg\}\boldsymbol{\xi}^{n}\\ \label{ADMBoundaction}
+&\int\limits_{\partial M}\Bigg\{ - \partial_c\left(\boldsymbol{\gamma}^{cd}\boldsymbol{\Pi}_{da}\right) -(\partial_a\boldsymbol{\gamma}^{cd})\boldsymbol{\Pi}_{cd} +\epsilon\partial_c \left(\boldsymbol{\xi}^{c}\boldsymbol{\varphi}_a\right) \Bigg\} \boldsymbol{\xi}^{a}.
\end{align}
\end{theorem}

\begin{proof}
Using the vertical vector fields in \eqref{kernelBVADM} to eliminate $\beta_a, \chi_\rho$ and $g^\dag{}^{ab}$ one is able to find an explicit section of the symplectic reduction $\pi:\widetilde{\mathcal{F}}_{ADM}\longrightarrow \mathcal{F}^\partial_{ADM}$. First of all, use $\mathbb{B}$ to set $\beta_a=0$, this implies $(X_\beta)_a=-\beta_a^0$ together with $\beta_a(t)=(1-t)\beta_a^0$ and we have the first two differential equations:
\begin{subequations}\begin{align}
\dot{\xi}^b &= +\gamma^{ab}\beta_a^0\xi^n\\
\dot{g}^\dag{}^{nn}&={ - }\frac{\epsilon}{2}\eta^{-2}\gamma^{ab}\beta_a^0\chi_b\xi^n
\end{align}\end{subequations}
that are easily solved to yield
\begin{subequations}\begin{align}
\xi(t)=&\xi_0^b + \gamma^{ba}\beta^0_a\xi^nt\\\label{g+nnpart}
g^\dag{}^{nn}(t)=&g_0^\dag{}^{nn} { - }\frac{\epsilon}{2}\eta^{-2}\gamma^{ab}\beta_a^0\chi_b\xi^n t
\end{align}\end{subequations}
we use \eqref{g+nnpart} and the time rule for $\beta_a(t)$ to solve equation
\begin{align*}
\dot{g}^\dag{}^{nb}=&{ + }\frac\epsilon2\eta^{-2} \beta^b\gamma^{cd}\beta^0_c\chi_d\xi^n + \gamma^{ab}\beta^0_ag^\dag{}^{nn}=\gamma^{ba}\beta^0_ag_0^\dag{}^{nn} { - }\frac{\epsilon}{2}\eta^{-2}\gamma^{ba}\beta^0_a\gamma^{cd}\beta_c^0\chi_d\xi^n(2t - 1)\\
{g}^\dag{}^{nb}(t)=&{g}_0^\dag{}^{nb} + \gamma^{ba}\beta^0_ag_0^\dag{}^{nn}t { - }\frac{\epsilon}{2}\eta^{-2}\gamma^{ba}\beta^0_a\gamma^{cd}\beta_c^0\chi_d\xi^n(t^2 - t)
\end{align*}
together with
\begin{align*}
\dot{J}_{lm}=&-\left(2\nabla_{(l}\beta^0_{m)} + \frac{4\epsilon}{\sqrt{\gamma}}\eta\left(\beta^0_{(l}\gamma_{m)a} - \frac{1}{d-1}\gamma_{lm}\beta^0_a \right)g^\dag{}^{an}\xi^n \right)\\
=&-2\nabla_{(l}\beta^0_{m)} - \frac{4\epsilon}{\sqrt{\gamma}}\eta\left(\beta^0_{(l}\gamma_{m)a} - \frac{1}{d-1}\gamma_{lm}\beta^0_a \right)g_0^\dag{}^{an}\xi^n + \\
&- \frac{2\epsilon}{\sqrt\gamma}\left(\beta^0_{(l}\beta^0_{m)} - \frac{1}{d-1}\gamma_{lm}\beta^0_b\beta_0^b\right)g^\dag_0{}^{nn}\xi^n t\\
{J}_{lm}=&-2\nabla_{(l}\beta^0_{m)}t - \frac{4\epsilon}{\sqrt{\gamma}}\eta\left(\beta^0_{(l}\gamma_{m)a} - \frac{1}{d-1}\gamma_{lm}\beta^0_a \right)g_0^\dag{}^{an}\xi^nt + \\
&- \frac{2\epsilon}{\sqrt\gamma}\left(\beta^0_{(l}\beta^0_{m)} - \frac{1}{d-1}\gamma_{lm}\beta^0_b\beta_0^b\right)g^\dag_0{}^{nn}\xi^n t^2
\end{align*}
So we can set the temporary value of our fields at $t=1$ to be
\begin{subequations}
\begin{align}\notag
\hat{J}_{lm}=&-2\nabla_{(l}\beta^0_{m)} - \frac{4\epsilon}{\sqrt{\gamma}}\eta\left(\beta^0_{(l}\gamma_{m)a} - \frac{1}{d-1}\gamma_{lm}\beta^0_a \right)g_0^\dag{}^{an}\xi^n + \\
&- \frac{2\epsilon}{\sqrt\gamma}\left(\beta^0_{(l}\beta^0_{m)} - \frac{1}{d-1}\gamma_{lm}\beta^0_b\beta_0^b\right)g^\dag_0{}^{nn}\xi^n \\
\hat{g}^\dag{}^{nb}=&g^\dag{}^{nb} + \gamma^{ba}\beta_ag^\dag{}^{nn}\\
\hat{g}^\dag{}^{nn}=&g^\dag{}^{nn} { - }\frac{\epsilon}{2}\eta^{-2} \gamma^{ab}\beta_a\chi_b\xi^n
\end{align}\end{subequations}

Now we can turn to the vector fields $\mathbb{X}_\rho$ and use them to set $\chi_\rho=0$ at some value of the internal evolution parameter $s$. As usual we impose $(X_\chi)_\rho=-\chi_\rho^0$ and $\chi_\rho(t)=(1-t)\chi_\rho^0$. The new equations are
\begin{align*}
\dot{g}^\dag{}^{nn}=&{ + }\frac{\epsilon}{2}\eta^{-2}\chi^0_n\xi^n \\
\dot{g}^\dag{}^{nb}=&{ - }\frac{\epsilon}{2}\gamma^{ba}\chi_a^0\xi^n
\end{align*}
which will yield an additional correction to the temporary value of our fields:
\begin{subequations}
\begin{align}
\hat{\hat{g}}^\dag{}^{nn}=&g^\dag{}^{nn} { + }\frac{\epsilon}{2}\eta^{-2}\left(\chi_n - \gamma^{ab}\beta_a\chi_b\right)\xi^n\\
\hat{\hat{g}}^\dag{}^{nb}=&g^\dag{}^{nb} + \gamma^{ba}\beta_ag^\dag{}^{nn} { - }\frac{\epsilon}{2}\gamma^{ba}\chi_a\xi^n
\end{align}\end{subequations}

Similar is what happens when we use $\mathbb{G}^\dag{}^{ab}$, for we get the equation
\begin{equation*}
\dot{J}_{lm}=-\frac{2\epsilon}{\sqrt{\gamma}}\eta\left(\gamma_{al}\gamma_{bm} - \frac{1}{d-1}\gamma_{ab}\gamma_{lm}\right)g_0^\dag{}^{ab}\xi^n
\end{equation*}
that will correct the temporary value of $\hat{J}_{lm}$ to
\begin{align}\notag
\hat{\hat{J}}_{lm}=&-2\nabla_{(l}\beta^0_{m)} - \frac{4\epsilon}{\sqrt{\gamma}}\eta\left(\beta^0_{(l}\gamma_{m)a} - \frac{1}{d-1}\gamma_{lm}\beta^0_a \right)g_0^\dag{}^{an}\xi^n + \\\notag
&- \frac{2\epsilon}{\sqrt\gamma}\left(\beta^0_{(l}\beta^0_{m)} - \frac{1}{d-1}\gamma_{lm}\beta^0_b\beta_0^b\right)g^\dag_0{}^{nn}\xi^n \\
&-\frac{2\epsilon}{\sqrt{\gamma}}\eta\left(\gamma_{al}\gamma_{bm} - \frac{1}{d-1}\gamma_{ab}\gamma_{lm}\right)g_0^\dag{}^{ab}\xi^n
\end{align}

Finally, we use the vector field $\mathbb{E}$ to set $\eta=1$. This implies that the time law for $\eta$ be given by $\eta(t)=(1-\eta_0)t + \eta_0$ and $(X_\eta)=1-\eta_0$. The associated equations read
\begin{align*}
\dot{\xi}^n =& -\frac{1-\eta_0}{(1-\eta_0)t + \eta_0}\xi^n\\
\dot{g}^\dag{}^{nn}=&-\frac{1-\eta_0}{(1-\eta_0)t + \eta_0}{g}^\dag{}^{nn}\\
\dot{J}_{lm}=&\frac{1-\eta_0}{(1-\eta_0)t+\eta_0}J_{lm}
\end{align*}
yielding, at time $t=1$, the following corrections to the fields: $\widetilde{\xi}^n=\eta\xi^n$, $\tgd{nn}=\eta \hat{\hat{g}}^\dag{}^{nn}$ and $\tJ{lm}=\eta^{-1}\hat{\hat{J}}_{lm}$. Putting everything together we get that the symplectic reduction map reads: 
\begin{equation}\pi\colon 
\begin{cases}
\tJ{lm}&=\eta^{-1}\left(J_{lm} - 2\nabla_{(l}\beta_{m)}\right) - \frac{2\epsilon}{\sqrt\gamma}\left(\gamma_{al}\gamma_{bm} - \frac{1}{d-1}\gamma_{lm}\gamma_{ab}\right)g^\dag{}^{ab}\xi^n\\
&-\frac{4\epsilon}{\sqrt\gamma}\left(\beta_{(l}\gamma_{m)b}-\frac{1}{d-1}\gamma_{lm}\beta_b\right)g^\dag{}^{bn}\xi^n - \frac{2\epsilon}{\sqrt\gamma}\left(\beta_{(l}\beta_{m)} - \frac{1}{d-1}\gamma_{lm}\beta_b\beta^b\right)g^\dag{}^{nn}\xi^n\\
\tgd{nn}&=\eta g^\dag{}^{nn} +\frac{\epsilon}{2}\eta^{-1}\left(\chi_n - \beta^a\chi_a \right)\xi^n \\
\tgd{bn}&=g^\dag{}^{bn} + \gamma^{ba}\beta_a g^\dag{}^{nn} + \frac\epsilon2\gamma^{ba}\chi_a\xi^n \\
\txi{b}&=\xi^b + \gamma^{ba}\beta_a\xi^n \\
\txi{n}&=\eta\xi^n\\
\tgam_{ab}&=\gamma_{ab}
\end{cases}
\end{equation}

The boundary $1$-form $\alpha^\partial$ will be given by the ansatz
\begin{equation}
\alpha^\partial=\epsilon\int\limits_{\partial M}\left\{\frac{\sqrt{\tgam}}{2}\left( \delta\tgam^{ab}\tgam_{ab}\tgam^{lm}\tJ{lm} - \delta \tgam^{lm}\tJ{lm} \right) - 2\delta\txi{n}\tgd{nn} + 2\gamma_{ab}\delta\txi{a}\tgd{bn}\right\}
\end{equation}
as it is straightforward to check that $\pi^*\alpha^\partial = \widetilde{\alpha}_{ADM}$. 

Introducing the new variables ${\boldsymbol{\gamma}}^{ab}\equiv\tgam^{ab}$, $\boldsymbol{\Pi}_{ab}=\frac{\sqrt{\tgam}}{2}\left(\tJ{ab} - \tgam_{ab}\tgam^{ij}\tJ{ij}\right)$ together with $\boldsymbol{\varphi}_{n}=-2\tgd{nn}$, $\boldsymbol{\varphi}_{a}=2\tgam_{ab}\tgd{bn}$ and $\boldsymbol{\xi}^\rho=\txi{\rho}$, we can write the symplectic boundary form as:
\begin{equation}
\omega^\partial=\epsilon\int\limits_{\partial M} \delta \boldsymbol{\gamma}^{ab}\delta\boldsymbol{\Pi}_{ab} + \delta\boldsymbol{\xi}^\rho\delta\boldsymbol{\varphi}_{\rho}
\end{equation}
and recover expression \eqref{projADM} and \eqref{ADMBoundaction} for the projection and the boundary action in the Darboux coordinates.

We would like to compute now the cohomological boundary vector field. First of all we must extract the analogous bulk vector field, encoding the equations of motion and the symmetries of the system, using the fundamental formula:
\begin{equation}\label{QADM}
\iota_Q\Omega_{BV} = \delta S + \pi_M^*\alpha^\partial
\end{equation}
A shortcut to do this in the ADM formalism, instead of computing cumbersome integrations by parts, consists in considering the classical Einstein-Hilbert action, whose classical vacuum equations of motion are given by 
$$\sqrt{\gamma}\left(R_{\mu\nu} - \left(\frac{1}{2}R-\Lambda\right)g_{\mu\nu}\right) \equiv G_{\mu\nu}=0$$
and to express them using the ADM  decomposition. This is done projecting the above equation on the new field direction, with the help of the Gauss-Codazzi equations and the Ricci equations.

Doing so, one obtains the projection of the relevant Euler-Lagrange terms in the ADM formalism, namely, forgetting about the BV extension for a moment:

\begin{align}\label{constraint1}
 G_\eta\coloneqq \pard{S_{ADM}}{\eta}&=\epsilon\sqrt{\gamma}\left( \epsilon\left(R^\partial - 2\Lambda\right) + K^2 - K_{ab}K^{ab}\right)\\\label{constraint2}
 G_{\beta_a}\coloneqq \pard{S_{ADM}}{\beta_b}&=2\epsilon\gamma^{ba}\left[\partial_c(\sqrt{\gamma} \gamma^{cd}K_{da})+ \frac12\partial_a\gamma^{cd}K_{cd}- \sqrt{\gamma}\partial_aK \right]\\
 G_{\gamma_{ab}}\coloneqq \pard{S_{ADM}}{\gamma_{ab}}&=\epsilon\sqrt{\gamma}\left(\partial_n K_{ab} - \beta^k\partial_kK_{ab} - 2K_{k(a}\partial_{b)}(g^{kc}\beta_c)\right)
\end{align}
Notice that the formula for $G_{\beta_a}$ is only apparently different from the usual \emph{momentum constraint} that can be found in the literature (see e.g. \cite{dWi}): 
$$\mathcal{H}_c\coloneqq \sqrt{\gamma}\gamma^{ba}\left(\gamma^{cd}\nabla^\partial_cK_{da} - \nabla^\partial_a K\right) $$
as  it can be seen by manipulating the covariant derivatives.

Adding the BV part we have that the derivatives of the action with respect to the new fields read:
\begin{align*}
\left(\pard{S_{BV}}{\eta}\right)&=-2\epsilon\eta\left( \partial_\rho\xi^\rho g^\dag{}^{nn} + \xi^\rho\partial_\rho g^\dag{}^{nn} - 2\partial_\rho\xi^ng^\dag{}^{n\rho}\right)\\
\left(\pard{S_{BV}}{\beta_a}\right)&=2\epsilon \left( \partial_\rho\xi^\rho g^\dag{}^{an} + \xi^\rho\partial_\rho g^\dag{}^{an} - \partial_\rho\xi^ag^\dag{}^{n\rho} - \partial_\rho\xi^n g^\dag{}^{a\rho}\right) + 2\epsilon\beta^a\left(\partial_\rho\xi^\rho g^\dag{}^{nn} + \xi^\rho\partial_\rho g^\dag{}^{nn} - 2\partial_\rho\xi^ng^\dag{}^{n\rho}\right)\\
\left(\pard{S_{BV}}{\gamma_{ab}}\right)&=\epsilon \left( \partial_\rho\xi^\rho g^\dag{}^{ab} + \xi^\rho\partial_\rho g^\dag{}^{ab} - 2\partial_\rho\xi^{(a}g^\dag{}^{b)\rho} \right)  - 2\epsilon\beta^a\beta^b\left(\partial_\rho\xi^\rho g^\dag{}^{nn} + \xi^\rho\partial_\rho g^\dag{}^{nn} - 2\partial_\rho\xi^ng^\dag{}^{n\rho}\right) \\ 
\left(\pard{S_{BV}}{g^\dag{}^{ab}}\right) &= \epsilon\left(\xi^\rho \partial_\rho \gamma_{ab} + 2\partial_{(a}\xi^n\beta_{b)} + 2\partial_{(a}\xi^c\gamma_{b)c}\right)\\
\left(\pard{S_{BV}}{g^\dag{}^{na}}\right) &=\epsilon\left( \xi^\rho\partial_\rho\beta_a + \partial_n\xi^n\beta_a +\partial_n\xi^b\gamma_{ab} + \partial_a\xi^n(-\eta^2 + \beta_c\beta^c) + \partial_a\xi^b\beta_b\right)\\
\left(\pard{S_{BV}}{g^\dag{}^{nn}}\right) &=\epsilon\left(\xi^\rho\partial_\rho(-\eta^2 + \beta_c\beta^c) + 2\partial_n\xi^n(-\eta^2 +\beta_c\beta^c) + 2\partial_n\xi^a \beta_a\right)
\end{align*}
together with
\begin{align*}
\left(\pard{S_{BV}}{\xi^n}\right) &=\epsilon\left(\partial_n(-\eta^2 +\beta_c\beta^c)g^\dag{}^{nn} + 2\partial_a(-\eta^2 + \beta_c\beta^c)g^\dag{}^{na} + 2\partial_{(a}\beta_{b)}g^\dag{}^{ab}\right) \\
&+\epsilon\left( 2(-\eta^2 + \beta_c\beta^c)\partial_ng^\dag{}^{nn} + 2\beta_a\partial_ng^\dag{}^{na} + 2(-\eta^2+\beta_c\beta^c)\partial_ag^\dag{}^{na}\right)\\
&+\epsilon\left(2\beta_{(a}\partial_{b)}g^\dag{}^{ab} - J_{ab}g^\dag{}^{ab}\right) + \xi^\rho\partial_\rho\chi_n + \partial_\rho\xi^\rho\chi_n + \partial_n\xi^\rho\chi_\rho\\
\left(\pard{S_{BV}}{\xi^a}\right) &= 2\epsilon\left(\partial_n\beta_ag^\dag{}^{nn} +  J_{ab}g^\dag{}^{nb} + \partial_b \beta_ag^\dag{}^{nb} +\partial_{(b}\gamma_{c)a}g^\dag{}^{bc} \right) \\ 
&+2\epsilon\left(\beta_a\partial_n g^\dag{}^{nn}\beta_a\partial_bg^\dag{}^{nb}+\gamma_{ab}\partial_ng^\dag{}^{nb} + \gamma_{a(b}\partial_{c)}g^\dag{}^{bc} - \frac12\partial_a\gamma_{cd}g^\dag{}^{cd}\right) \\
&- \epsilon \partial_a\left(-\eta^2 +\beta_c\beta^c\right)g^\dag{}^{nn} - 2\epsilon\partial_a\beta_ccg^\dag{}^{cn}+\partial_\rho\xi^\rho\chi_a + \xi^\rho\partial_\rho\chi_a + \partial_a\xi^\rho\chi_\rho \\
\left(\pard{S_{BV}}{\chi_\mu}\right) &=\xi^\rho\partial_\rho\xi^\mu
\end{align*}

Now we would like to use these derivatives to write down the components of the bulk vector field $Q$, by  imposing \eqref{QADM}. We are using the antifields coordinates $g^\dag{}^{\mu\nu}$, conjugate to $g_{\mu\nu}$, and therefore we have to expand and reorder the terms as follows:
\begin{align*}
(Q_{g^\dag})^{nn} =& -\frac12\eta^{-1}\epsilon \left(\pard{S}{\eta}\right) = -\frac12\eta^{-1}\epsilon G_\eta +\left( \partial_\rho\xi^\rho g^\dag{}^{nn} + \xi^\rho\partial_\rho g^\dag{}^{nn} - 2\partial_\rho\xi^ng^\dag{}^{n\rho}\right)\\
(Q_{g^\dag})^{na}=& \frac\epsilon2 \left(\pard{S}{\beta_a}\right) -  \beta^a (Q_{g^\dag})^{nn} = \epsilon G_{\beta_a}+\frac\epsilon2 \eta^{-1}\beta^aG_\eta + \left( \partial_c\xi^c g^\dag{}^{an} + \xi^\rho\partial_\rho g^\dag{}^{an} - \partial_\rho\xi^ag^\dag{}^{n\rho} - \partial_c\xi^n g^\dag{}^{ac}\right) \\ 
(Q_{g^\dag})^{ab}=&\epsilon\left(\pard{S}{\gamma_{ab}}\right) + \beta^a\beta^a(Q_{g^\dag})^{nn} = \epsilon G_{\gamma_{ab}} -\frac\epsilon2\eta^{-1}\beta^a\beta^bG_\eta + \left( \partial_\rho\xi^\rho g^\dag{}^{ab} + \xi^\rho\partial_\rho g^\dag{}^{ab} - 2\partial_\rho\xi^{(a}g^\dag{}^{b)\rho} \right)
\end{align*}
together with
\begin{align*}
(Q_\gamma)_{ab} =& \epsilon \left(\pard{S_{BV}}{g^\dag{}^{ab}}\right) =\left(\xi^\rho \partial_\rho \gamma_{ab} + 2\partial_{(a}\xi^n\beta_{b)} + 2\partial_{(a}\xi^c\gamma_{b)c}\right)\\
(Q_\beta)_{a} =& \epsilon \left(\pard{S_{BV}}{g^\dag{}^{an}}\right) =\left( \xi^\rho\partial_\rho\beta_a + \partial_n\xi^n\beta_a +\partial_n\xi^b\gamma_{ab} + \partial_a\xi^n(-\eta^2 + \beta_c\beta^c) + \partial_a\xi^b\beta_b\right)\\
(Q_\eta) =& - \frac{\epsilon}{2}\eta^{-1}  \left(\pard{S_{BV}}{g^\dag{}^{nn}}\right)  + \epsilon\eta^{-1}\beta^a\left(\pard{S_{BV}}{g^\dag{}^{na}}\right) -\frac{\epsilon}{2}\eta^{-1}\beta^a\beta^b \left(\pard{S_{BV}}{g^\dag{}^{ab}}\right)= \left( \xi^\rho\partial_\rho\eta + \partial_n\xi^n \eta - \eta\beta^a\partial_a\xi^n\right)\\
\end{align*}
and, with $\rho=1,2,3,n$:
\begin{equation*}
(Q_\xi)^\rho = \left(\pard{S_{BV}}{\chi_\rho}\right);\ \ \ 
(Q_\chi)_\rho = \left(\pard{S_{BV}}{\xi^{\rho}}\right) 
\end{equation*}

The bulk $Q$ vector field is extended to the normal jets when projected to the pre-boundary vector field $\widetilde{Q}$:
\[
(\widetilde{Q}J)_{\mu\nu}=(\partial_n(Qg_{\mu\nu}))\big|_{\partial M}=\left(\partial_n\xi^\rho\partial_\rho g_{\mu\nu} + \xi^\rho\partial_\rho \partial_ng_{\mu\nu} + 2\partial_{(\mu}\partial_n\xi^\rho g_{\nu)\rho} + 2\partial_{(\mu}\xi^\rho \partial_ng_{\nu)\rho}\right)\big|_{\partial M}
\]
of which we will only need 
\[
(\widetilde{Q}J)_{ab}=\epsilon\left(\partial_n\xi^\rho\partial_\rho \gamma_{ab} + \xi^\rho\partial_\rho J_{ab} + 2\partial_{(a}\partial_n\xi^\rho g_{b)\rho} + 2\partial_{(a}\xi^c J_{b)c} + 2\partial_{(a}\xi^n \partial_n\beta_{b)}\right)
\]
so that the full pre-boundary vector field reads:
\begin{align*}
\widetilde{Q}=&(\widetilde{Q}_\eta)\pard{}{\eta} + (\widetilde{Q}_\beta)_a\pard{}{\beta_a} + (\widetilde{Q}_\gamma)_{ab}\pard{}{\gamma_{ab}}+(\widetilde{Q}_{g^\dag})^{ab}\pard{}{g^\dag{}^{ab}} +2(\widetilde{Q}_{g^\dag})^{na}\pard{}{g^\dag{}^{na}} \\
+&(\widetilde{Q}_{g^\dag})^{nn}\pard{}{g^\dag{}^{nn}} + (\widetilde{Q}_\xi)^n\pard{}{\xi^n} + (\widetilde{Q}_\xi)^a\pard{}{\xi^a} + (\widetilde{Q}_\chi)_\mu\pard{}{\chi_\mu} + (\widetilde{Q}_J)_{lm}\pard{}{J_{lm}}
\end{align*}

Instead of computing the explicit projection of $\widetilde{Q}$ directly, we consider the following simplifying technique. We produce a degree one function via \cite{Royt}:
\[\widetilde{S}=\iota_{\widetilde{Q}}\iota_{\widetilde{E}}\widetilde{\omega}\]
where $\widetilde{E}$ is the Euler vector field on the space of pre-boundary fields, i.e.:
\[\widetilde{E}=\int\limits_{\partial{M}}\xi^\rho\pard{}{\xi^\rho} - g^\dag{}^{\mu\nu}\pard{}{g^\dag{}^{\mu\nu}} - 2\chi_\rho\pard{}{\chi_\rho}\]
Then the true boundary action $S^\partial$ is such that $\widetilde{S}=\pi^*S^\partial$ for degree reasons and the surjectivity of the surjection $\pi_M$, which factors through $\widetilde{\mathcal{F}}_{ADM}$. Moreover $Q^\partial$ is its Hamiltonian vector field. The boundary action is then found to be:
\begin{align}\notag
S^\partial=&\int\limits_{\partial M}\Bigg\{\sqrt{\tgam}\left(\frac\epsilon4\left(\widetilde{J}^{ab}\tJ{ab} - \tJ{}^2\right) + R^\partial -\Lambda \right) - 2\epsilon\partial_a\left(\txi{a}\tgd{nn}\right) - 2\epsilon\tgd{na}\partial_a\txi{n}\Bigg\}\txi{n}\\
+&\int\limits_{\partial M}\Bigg\{\sqrt{\tgam}\partial_a\tJ{} - \partial_c\left(\sqrt{\tgam}\tgam^{cd}\tJ{da}\right) - \frac{\sqrt{\tgam}}{2}(\partial_a\tgam^{cd})\tJ{cd} +2\epsilon\partial_c\left(\txi{c}\tgd{nb}\tgam_{ba}\right)\Bigg\}\txi{a}
\end{align}
where by $\tJ{}$ we denote the trace $\tgam^{ab}\tJ{ab}$.

Again using the definition of the Darboux coordinates, we can easily gather that the boundary action in that local chart will yield the expression \eqref{ADMBoundaction}, whereas the explicit components of $Q^\partial$ can be found using the equation $\iota_{Q^\partial}\omega^\partial = \delta S^\partial$.
\end{proof}

This result is a clean first step in the direction of  BV-BFV quantisation of General Relativity as proposed by CMR in \cite{CMR3}. It states the compatibility of bulk and boundary structures, in relation with the symmetries. Notice that the BV-BFV axioms \ref{BV-BFV} need not be satisfied by a generic gauge theory and the statement is therefore nontrivial. Arguable as it might be to consider gauge theories with this property to be somehow \emph{better quantisable}, it provides nevertheless a clear mean of distinction between different variational problems describing the same equations of motion (see \cite{CS2} for a comparison with the Palatini-Holst formulation of GR).

The machinery is able to handle a more complex and sophisticated set of data, than the usual procedures of canonical analysis. When inducing (or not) a theory on the boundary it encodes a number of characteristic features, packing up relevant data in a  very efficient way. As we will see in the following section, the piece of data that carries all the relevant information on the boundary is, not surprisingly, the boundary action.

Finally, recall that in the $1+1$ dimensional case it is known that the Einstein equations are trivial, and the symmetry distribution has to be amended to take conformal transformations into account. The critical dimension $d=1$ is however marked out by the equations for the kernel of the pre-boundary $2$-form $\widetilde{\omega}$, both in the classical and the BV-extended case (cf. Theorem \ref{Theo:ADM1} and Proposition \ref{Prop:ClAdm}), confirming that the strategy has to be altered to analyse this specific example.


\subsection{Constraint algebra and boundary gauge symmetry}
As we already announced, from the boundary action \eqref{ADMBoundaction} it is possible to read the constraint structure of canonical gravity. As a matter of fact, the degree zero (ghost number, $gh$) part of the derivatives $\pard{S^\partial}{\xi^\mu}$ reads
\begin{align}
\pard{S^\partial}{\xi^n}\Bigg|_{gh=0}=&\frac{\epsilon}{\sqrt{\boldsymbol{\gamma}}}\left(\boldsymbol{\Pi}^{ab}\boldsymbol{\Pi}_{ab} - \frac{1}{d-1}\boldsymbol{\Pi}^2\right) + \sqrt{\boldsymbol{\gamma}}\left(R^\partial -2\Lambda\right)\equiv \mathcal{H}\\
\pard{S^\partial}{\xi^a}\Bigg|_{gh=0}=&- \partial_c\left(\boldsymbol{\gamma}^{cd}\boldsymbol{\Pi}_{da}\right) -(\partial_a\boldsymbol{\gamma}^{cd})\boldsymbol{\Pi}_{cd}\equiv\mathcal{H}_a
\end{align}
which are the symplectic-reduced versions of the standard constraints \eqref{constraint1} and \eqref{constraint2}.

On the other hand, the residual gauge symmetries can be found by computing the relative components of the boundary cohomological vector field $Q^\partial$, using the fact that $\iota_{Q^\partial}\omega^\partial=\delta S^\partial$:
\begin{align*}
(Q^\partial)_{\boldsymbol{\xi}^n}=&\left( \boldsymbol{\xi}^c\partial_c\boldsymbol{\xi}^n\right)\pard{}{\boldsymbol{\xi}^n}\\
(Q^\partial)_{\boldsymbol{\xi}^{a}}=&\left(\boldsymbol{\xi}^n\boldsymbol{\gamma}^{ab}\partial_b\boldsymbol{\xi}^n + \boldsymbol{\xi}^c\partial_c\boldsymbol{\xi}^{a}\right)\pard{}{\boldsymbol{\xi}^{a}}\\
(Q^\partial)_{\boldsymbol{\gamma}_{ab}}=&\left(\boldsymbol{\xi}^n\frac{2}{\sqrt{\boldsymbol{\gamma}}}(\boldsymbol{\Pi}_{ab} - \frac{\boldsymbol{\gamma}_{ab}}{d-1}\boldsymbol{\Pi}) +\boldsymbol{\xi}^c\partial_c\boldsymbol{\gamma}_{ab}+2\partial_{(a}\boldsymbol{\xi}^c\boldsymbol{\gamma}_{b)c}\right)\pard{}{\boldsymbol{\gamma}_{ab}}
\end{align*}

It is interesting to notice that the symmetries above are a corrected version of the usual gauge symmetry for a $d$-dimensional metric on the boundary under the action of boundary diffeomorphisms $\boldsymbol{\xi}^\partial\in T[1]\partial M$. In fact they can be compactly rewritten as

\begin{align}
(Q^\partial)_{\boldsymbol{\gamma}}=&\boldsymbol{\xi}^n\frac{2}{\sqrt{\boldsymbol{\gamma}}}(\boldsymbol{\Pi} - \frac{\boldsymbol{\gamma}}{d-1}\mathrm{Tr}\boldsymbol{\Pi}) + L_{\boldsymbol{\xi}^\partial}\boldsymbol{\gamma}\\
(Q^\partial)_{\boldsymbol{\xi}^\partial}=&\boldsymbol{\xi}^n \boldsymbol{\gamma}^{-1}\nabla\boldsymbol{\xi}^n + \frac12[\boldsymbol{\xi}^\partial,\boldsymbol{\xi}^\partial]\\
(Q^\partial)_{\boldsymbol{\xi}^n}=&L_{\boldsymbol{\xi}^\partial}\xi^n
\end{align}

This means that they do not manifestly show a Lie algebra behaviour and the \emph{structure functions} depend on $\boldsymbol{\gamma}^{-1}$. Yet the boundary BFV action \eqref{ADMBoundaction} is at most linear in the antighosts $\boldsymbol\varphi$. This is in agreement with the observations in \cite{BlohWein}. The BFV formalism provides for a cohomological resolution of symmetry-invariant coisotropic submanifolds \cite{Schaetz09,Schaetz10,BFV1,Stash,stash97}, and in this case of the constraint submanifold of canonical gravity, modulo residual gauge symmetry. 

The (cohomological) description of the the canonical, constrained phase space for General Relativity is then obtained from a simple variational problem in the bulk. This encompasses a number of classical results in the field while clarifying related issues at the same time. Moreover, we stress that on top of obtaining the expected BFV resolution of the canonical structure on the boundary, we are able to establish a connection with the boundary data through the explicit projection $\pi$, and the fundamental equation $\iota_Q\Omega=\delta S_{ADM}^{BV}\alpha^\partial$. This is the starting point for the BV-BFV programme to quantisation of gauge theories on manifolds with boundary.

\end{document}